%% file: main.tex
\documentclass[10pt,conference,romanappendices]{IEEEtran}
\input{commands}
\usepackage[cmex10]{amsmath}
\usepackage{amsfonts}
\usepackage{wrapfig}
\usepackage{graphicx}
\usepackage{enumerate}
\usepackage{multirow}
\usepackage{xspace}
\usepackage{cancel}
\usepackage{underoverlap}
\begin{document}

\title{Complete Test Sets And Their Approximations}


\author{
\IEEEauthorblockN{Eugene Goldberg}
\IEEEauthorblockN{\emph{eu.goldberg@gmail.com}}}


\maketitle

\input{abstract}
\input{i1ntroduction}

\input{s1sa}

\input{c1ts}
\input{a1lgor}

\input{i2nternal_cut}

\input{a2pplications}

\input{e4xper}

\input{b2ackground}

\input{c6onclusion}
\clearpage
\bibliographystyle{plain}
\bibliography{short_sat,local}
\input{a4ppendix}

\end{document}

%% file: commands.tex
\newtheorem{example}{Example}
\newtheorem{definition}{Definition}
\newtheorem{proposition}{Proposition}

\newenvironment{proof}{\hspace{8pt}\ti{Proof:}}{~~~~QED}

\newcommand{\pnt}[1]{{\mbox{$\vec{#1}$}}}
\newcommand{\ppnt}[2]{{\mbox{$\vec{#1}_{#2}$}}}
\newcommand{\bm}[1]{{\mbox{\boldmath $#1$}}}
\newcommand{\cof}[2]{\mbox{$#1_{\vec{#2}}$}}

\newcommand{\V}[1]{\mbox{$\mathit{Vars}(#1)$}}

\newcommand{\s}[1]{\mbox{$\{#1\}$}}

\newcommand{\nGz}[2]{$G_{non-\{z\}}$}

\newcommand{\prr}[1]{\mi{Prev}(\boldsymbol{q})}

\newcommand{\mi}[1]{\mathit{#1}}
\newcommand{\ti}[1]{\textit{#1}}
\newcommand{\tb}[1]{\textbf{#1}}

\newcommand{\ttt}{\>\>\>}
\newcommand{\tttt}{\>\>\>\>}

\newcommand{\Tt}{\>\>}
\newcommand{\Sub}[2]{\mbox{$\mi{#1}_\mi{#2}$}}
\newcommand{\sub}[2]{\mbox{$\vec{#1}_\mi{#2}$}}

\newcommand{\Comment}[1]{}

\newcommand{\sas}{\mbox{$\mi{SemStr}$}\xspace}
\newcommand{\pct}{\ti{GenPCT}\xspace}
\newcommand{\cube}[1]{\mbox{$\mi{Cube}(\pnt{#1})$}}
\newcommand{\Cube}[2]{\mbox{$\mi{Cube}(\ppnt{#1}{#2})$}}
\newcommand{\nbhd}[2]{\mbox{$\mi{Nbhd}(\pnt{#1},#2)$}}
\newcommand{\nnbhd}[3]{\mbox{$\mi{Nbhd}(\ppnt{#1}{#2},#3)$}}
\newcommand{\Nbhd}[3]{\mbox{$\mi{Nbhd}(\sub{#1}{init},\pnt{#2},#3)$}}
\newcommand{\NNbhd}[4]{\mbox{$\mi{Nbhd}(\sub{#1}{init},\ppnt{#2}{#3},#4)$}}
\newcommand{\Nbh}[3]{\mbox{$\mi{Nbhd}(\pnt{#1},\pnt{#2},#3)$}}
\newcommand{\ac}[1]{\mbox{$\Phi(#1)$}}
\newcommand{\Fi}{\mbox{$\Phi$}\xspace}
\newcommand{\acts}{CTS\textsuperscript{a}\xspace}
\newcommand{\aacts}{CTS\textsuperscript{aa}\xspace}
\newcommand{\pim}{\mbox{$\xi(M,M^{\pi})$}\xspace}
\newcommand{\spi}[1]{\mbox{$#1^{\pi}$}}

%% file: abstract.tex
\begin{abstract}
  We use testing to check if a combinational circuit $N$ always
  evaluates to 0 (written as $N \equiv 0$).  We call a set of tests
  proving $N \equiv 0$ a complete test set (CTS). The conventional
  point of view is that to prove $N \equiv 0$ one has to generate a
  \ti{trivial} CTS. It consists of all $2^{|X|}$ input assignments
  where $X$ is the set of input variables of $N$. We use the notion of
  a Stable Set of Assignments (SSA) to show that one can build a
  \ti{non-trivial} CTS consisting of less than $2^{|X|}$ tests. Given
  an unsatisfiable CNF formula $H(W)$, an SSA of $H$ is a set of
  assignments to $W$ that proves unsatisfiability of $H$. A trivial
  SSA is the set of all $2^{|W|}$ assignments to $W$. Importantly,
  real-life formulas can have non-trivial SSAs that are much smaller
  than $2^{|W|}$.  In general, construction of even non-trivial CTSs
  is inefficient. We describe a much more efficient approach where
  tests are extracted from an SSA built for a ``projection'' of $N$ on
  a subset of variables of $N$.  These tests can be viewed as an
  approximation of a CTS for $N$. We give experimental results and
  describe potential applications of this approach.
 
\end{abstract}


%% file: i1ntroduction.tex
\section{Introduction}
  Testing is an important part of verification flows. For that reason,
  any progress in understanding testing and improving its quality is
  of great importance. In this paper, we consider the following
  problem. Given a single-output combinational circuit $N$, find a set
  of input assignments (tests) proving that $N$ evaluates to 0 for
  every test (written as $N \equiv 0$) or find a counterexample. We
  will call a set of input assignments proving $N \equiv 0$ a
  \ti{complete test set} (\ti{CTS})\footnote{\input{f7ootnote}}.  We
  will call the set of all possible tests a \ti{trivial
    CTS}. Typically, one assumes that proving $N \equiv 0$ involves
  derivation of the trivial CTS, which is infeasible in practice.
  Thus, testing is used only for finding an input assignment refuting
  $N \equiv 0$. We present an approach for building a non-trivial CTS
  consisting only of a subset of all possible tests. In general,
  finding even a non-trivial CTS for a large circuit is
  impractical. We describe a much more efficient approach where an
  \ti{approximation} of a CTS is generated.

  The circuit $N$ above usually describes a property $\xi$ of a
  multi-output combinational circuit $M$, the latter being the
  \ti{real object of testing}.  For instance, $\xi$ may state that $M$
  never produces some output assignments. To differentiate CTSs and
  their approximations from conventional test sets verifying $M$ ``as
  a whole'', we will refer to the former as \ti{property-checking test
    sets}.  Let $\Xi :=\s{\xi_1,\dots,\xi_k}$ be the set of properties
  of $M$ formulated by a designer. Assume that every property of $\Xi$
  holds and $T_i$ is a test set generated to check property $\xi_i \in
  \Xi$.  There are at least two reasons why applying $T_i$ to $M$
  makes sense. First, if $\Xi$ is \ti{incomplete}\footnote{That is $M$
    can be incorrect even if all properties of $\Xi$ hold.}, a test of
  $T_i$ can expose a bug, if any, breaking a property of $M$ that is
  not in $\Xi$. Second, if property $\xi_i$ is defined
  \ti{incorrectly}, a test of $T_i$ may expose a bug breaking the
  correct version of $\xi_i$.  On the other hand, if $M$ produces
  proper output assignments for all tests of $T_1 \cup \dots \cup
  T_k$, one gets extra guarantee that $M$ is correct.  In
  Section~\ref{sec:appl}, we list some other applications of
  property-checking test sets such as verification of design changes,
  hitting corner cases and testing sequential circuits.

  Let $N(X,Y,z)$ be a single-output combinational circuit where $X$
  and $Y$ specify the sets of input and internal variables of $N$
  respectively and $z$ specifies the output variable of $N$. Let
  $F_N(X,Y,z)$ be a formula defining the functionality of $N$ (see
  Section~\ref{sec:cts}). We will denote the set of variables of
  circuit $N$ (respectively formula $H$) as \V{N} (respectively
  \V{H}). Every assignment\footnote{\input{f1ootnote}} to \V{F_N}
  satisfying $F_N$ corresponds to a consistent
  assignment\footnote{\input{f2ootnote}} to \V{N} and vice versa. Then
  the problem of proving $N \equiv 0$ reduces to showing that formula
  $F_N \wedge z$ is unsatisfiable.  From now on, we assume that all
  formulas mentioned in this paper are \ti{propositional}. Besides, we
  will assume that every formula is represented in CNF i.e. as a
  conjunction of disjunctions of literals.

  Our approach is based on the notion of a Stable Set of Assignments
  (SSA) introduced in~\cite{ssp}.  Given formula $H(W)$, an SSA of $H$
  is a set $P$ of assignments to variables of $W$ that have two
  properties.  First, every assignment of $P$ falsifies $H$. Second,
  $P$ is a transitive closure of some neighborhood relation between
  assignments (see Section~\ref{sec:ssa}). The fact that $H$ has an
  SSA means that the former is unsatisfiable. Otherwise, an assignment
  satisfying $H$ is generated when building its SSA. If $H$ is
  unsatisfiable, the set of all $2^{|W|}$ assignments is always an SSA
  of $H$ . We will refer to it as \ti{trivial}. Importantly, a
  real-life formula $H$ can have a lot of SSAs whose size is much less
  than $2^{|W|}$. We will refer to them as \ti{non-trivial}.  As we
  show in Section~\ref{sec:ssa}, the fact that $P$ is an SSA of $H$ is
  a \ti{structural} property of the latter. That is this property
  cannot be expressed in terms of the truth table of $H$ (as opposed
  to a \ti{semantic} property of $H$). For that reason, if $P$ is an
  SSA for $H$, it may not be an SSA for some other formula $H'$ that
  is logically equivalent to $H$. In other words, a structural
  property is \ti{formula-specific}.

  We show that a CTS for $N$ can be easily extracted from an SSA of
  formula $F_N \wedge z$. This makes a non-trivial CTS a structural
  property of circuit $N$ that cannot be expressed in terms of its
  truth table.  Building an SSA for a large formula is inefficient.
  So, we present a procedure constructing a simpler formula $H(V)$
  implied by $F_N \wedge z$ $($where $V \subseteq \V{F_N \wedge z})$ and
  building an SSA of $H$.  The existence of such an SSA means that $H$
  (and hence $F_N \wedge z$) is unsatisfiable. So, $N \equiv 0$ holds.
  A test set extracted from an SSA of $H$ can be viewed as a way to
  verify a ``projection'' of $N$ on variables of $V$. On the other
  hand, one can consider this set as an approximation of a CTS for
  $N$.
  We will refer to the procedure above as \sas (``\ti{Sem}antics and
  \ti{Str}ucture''). \sas combines semantic and structural
  derivations, hence the name. The semantic part of \sas
  is\footnote{\input{f6ootnote}} to derive $H$. Its structural part
  consists of constructing an SSA of $H$ thus proving that $H$ is
  unsatisfiable.

  The contribution of this paper is fourfold. First, we introduce the
  notion of non-trivial CTSs (Section~\ref{sec:cts}). Second, we
  present a method for efficient construction of property-checking
  tests that are approximations of CTSs (Sections~\ref{sec:algor}
  and~\ref{sec:app_cts}).  Third, we describe applications of such
  tests (Section~\ref{sec:appl}).  Fourth, we give experimental
  results showing the effectiveness of property-checking tests
  (Section~\ref{sec:exper}).

%% file: f7ootnote.tex
Term CTS is sometimes used to say that a test set invokes every event
specified by a \ti{coverage metric}. Our application of this term is
quite different.

%% file: f1ootnote.tex
By an assignment to a set of variables $V$, we mean a \ti{full}
assignment where every variable of $V$ is assigned a value.

%% file: f2ootnote.tex
An assignment to a gate $G$ of $N$ is called consistent if the value
assigned to the output variable of $G$ is implied by values assigned
to its input variables. An assignment to variables of $N$ is called
consistent if it is consistent for every gate of $N$.

%% file: f6ootnote.tex
Implication $F_N \wedge z \rightarrow H$ is a \ti{semantic} property
of $F_N \wedge z$.  To verify this property it suffices to know the
truth table of $F_N \wedge z$.

%% file: s1sa.tex
\section{Stable Set Of Assignments}
\label{sec:ssa}
\subsection{Definitions}
We will refer to a disjunction of literals as a \ti{clause}.  Let
\pnt{p}\, be an assignment to a set of variables $V$. Let \pnt{p}\,
falsify a clause $C$.  Denote by {\boldmath \nbhd{p}{C}} the set of
assignments to $V$ satisfying $C$ that are at Hamming distance 1 from
\pnt{p}. (Here \ti{Nbhd} stands for ``Neighborhood''). Thus, the
number of assignments in \nbhd{p}{C} is equal to that of literals in
$C$. Let \pnt{q}\, be another assignment to $V$ (that may be equal to
\pnt{p}). Denote by {\boldmath \Nbh{q}{p}{C}} the subset of
\nbhd{p}{C} consisting only of assignments that are farther away from
\pnt{q} than \pnt{p} (in terms of the Hamming distance).

\begin{example}
  Let $V=\s{v_1,v_2,v_3,v_4}$ and \pnt{p}=0110. We assume that the
  values are listed in \pnt{p} in the order the corresponding
  variables are numbered i.e. \mbox{$v_1=0$}, $v_2=1,v_3=1,v_4=0$. Let $C= v_1
  \vee \overline{v_3}$. (Note that \pnt{p} falsifies $C$.) Then
  \nbhd{p}{C}=\s{\ppnt{p}{1},\ppnt{p}{2}} where \ppnt{p}{1} = 1110 and
  \ppnt{p}{2}=0100. Let \pnt{q} = 0000. Note that \ppnt{p}{2} is
  actually closer to \pnt{q} than \pnt{p}. So
  \Nbh{q}{p}{C}=\s{\ppnt{p}{1}}.
\end{example}
\begin{definition}
  \label{def:ac_fun}
  Let $H$ be a formula\footnote{\input{f3ootnote}} specified by a set
  of clauses \s{C_1,\dots,C_k}.  Let $P$ =
  \s{\ppnt{p}{1},\dots,\ppnt{p}{m}} be a set of assignments to \V{H}
  such that every $\ppnt{p}{i} \in P$ falsifies $H$.  Let \Fi denote a
  mapping $P \rightarrow H$ where \ac{\ppnt{p}{i}} is a clause $C$ of
  $H$ falsified by \ppnt{p}{i}. We will call \Fi an \tb{AC-mapping}
  where ``AC'' stands for ``Assignment-to-Clause''. We will denote the
  range of \Fi as \ac{P}. (So, a clause $C$ of $H$ is in \ac{P} iff
  there is an assignment $\ppnt{p}{i} \in P$ such that $C =
  \Fi(\ppnt{p}{i})$.)
\end{definition}
\begin{definition}
 \label{def:ssa}
Let $H$ be a formula specified by a set of clauses
\s{C_1,\dots,C_k}. Let $P$ = \s{\ppnt{p}{1},\dots,\ppnt{p}{m}} be a
set of assignments to \V{H}. $P$ is called a \tb{Stable Set of
  Assignments}\footnote{\input{f5ootnote}} (\tb{SSA}) of $H$ with
\tb{center} $\sub{p}{init} \in P$ if there is an AC-mapping \Fi such
that for every $\ppnt{p}{i}\in P$, $\NNbhd{p}{p}{i}{C} \subseteq P$
holds where $C = \ac{\ppnt{p}{i}}$.
\end{definition}

\begin{example}
 \label{exmp:ssa}
  Let $H$ consist of four clauses: $C_1 = v_1 \vee v_2 \vee v_3$, $C_2
  = \overline{v}_1$, $C_3 = \overline{v}_2$, $C_4 = \overline{v}_3$.
  Let $P =\s{\ppnt{p}{1},\ppnt{p}{2},\ppnt{p}{3},\ppnt{p}{4}}$ where
  $\ppnt{p}{1} = 000$, $\ppnt{p}{2} = 100$, $\ppnt{p}{3} = 010$,
  $\ppnt{p}{4}=001$.  Let \Fi be an AC-mapping specified as
  $\ac{\ppnt{p}{i}} = C_i, i = 1,\dots,4$.  Since $\ppnt{p}{i}$
  falsifies $C_i$, $i=1,\dots,4$,~~\Fi is a correct AC-mapping. $P$ is
  an SSA of $H$ with respect to \Fi and center
  \sub{p}{init}=\ppnt{p}{1}. Indeed,
  \NNbhd{p}{p}{1}{C_1}=\s{\ppnt{p}{2},\ppnt{p}{3},\ppnt{p}{4}} where
  $C_1 = \ac{\ppnt{p}{1}}$ and \NNbhd{p}{p}{i}{C_i} = $\emptyset$,
  where $C_i = \ac{\ppnt{p}{i}}$, $i=2,3,4$. Thus,
  $\mi{Nbhd}(\sub{p}{init},\ppnt{p}{i},\ac{\ppnt{p}{i}}) \subseteq P$,
  $i=1,\dots,4$.
\end{example}
\subsection{SSAs and  satisfiability of a formula}
\label{ssec:ssa_sat}
\begin{proposition}
\label{prop:ssa}
  Formula $H$ is unsatisfiable iff it has an SSA.
\end{proposition}

The proof\footnote{The proof of Proposition~\ref{prop:ssa} presented
  in report~\cite{cmpl_tst} is inacurate.} is given
Appendix~\ref{app:proofs}.  A similar proposition is proved
in~\cite{ssp} for ``uncentered'' SSAs (see Footnote~\ref{ftn:ssa}).

%
\input{b0uild_path.fig}

The set of all assignments to \V{H} forms the \ti{trivial} uncentered
SSA of $H$. Example~\ref{exmp:ssa} shows a \ti{non-trivial} SSA. The
fact that formula $H$ has a non-trivial SSA $P$ is its \ti{structural}
property. That is one cannot check whether $P$ is an SSA of $H$ if
only the truth table of $H$ is known. In particular, $P$ may not be an
SSA of a formula $H'$ logically equivalent to $H$.

\input{b1uild_ssa.fig}

The relation between SSAs and satisfiability can be explained as
follows. Suppose that formula $H$ is satisfiable. Let \sub{p}{init} be
an arbitrary assignment to \V{H} and \pnt{s} be a satisfying
assignment that is the closest to \sub{p}{init} in terms of the
Hamming distance.  Let $P$ be the set of all assignments to \V{H} that
falsify $H$ and \Fi be an AC-mapping from $P$ to $H$.  Then \pnt{s}
can be reached from \sub{p}{init} by procedure \ti{BuildPath} shown in
Figure~\ref{fig:bld_path}.  It generates a sequence of assignments
$\ppnt{p}{1},\dots,\ppnt{p}{i}$ where \ppnt{p}{1} = \sub{p}{init} and
\ppnt{p}{i}=\pnt{s}. First, \ti{BuildPath} checks if current
assignment \ppnt{p}{i} equals \pnt{s}. If so, then \pnt{s} has been
reached.  Otherwise, \ti{BuildPath} uses clause $C=\ac{\ppnt{p}{i}}$
to generate next assignment. Since \pnt{s} satisfies $C$, there is a
variable $v \in \V{C}$ that is assigned differently in \ppnt{p}{i} and
\pnt{s}. \ti{BuildPath} generates a new assignment \ppnt{p}{i+1}
obtained from \ppnt{p}{i} by flipping the value of $v$.

\ti{BuildPath} reaches \pnt{s} in $k$ steps where $k$ is the Hamming
distance between \sub{p}{init} and \pnt{s}.  Importantly,
\ti{BuildPath} reaches \pnt{s} for \ti{any} AC-mapping. Let $P$ be an
SSA of $H$ with respect to center \sub{p}{init} and AC-mapping
\Fi. Then if \ti{BuildPath} starts with \sub{p}{init} and uses \Fi as
an AC-mapping, it can reach only assignments of $P$. Since every
assignment of $P$ falsifies $H$, no satisfying assignment can be
reached.

A procedure for generation of SSAs called \ti{BuildSSA} is shown in
Figure~\ref{fig:bld_ssa}. It accepts formula $H$ and outputs either a
satisfying assignment or an SSA of $H$, center \sub{p}{init} and
AC-mapping \Fi. \ti{BuildSSA} maintains two sets of assignments
denoted as $E$ and $Q$.  Set $E$ contains the examined assignments
i.e. those whose neighborhood is already explored.  Set $Q$ specifies
assignments that are queued to be examined. $Q$ is initialized with an
assignment \sub{p}{init} and $E$ is originally empty. \ti{BuildSSA}
updates $E$ and $Q$ in a \ti{while} loop. First, \ti{BuildSSA} picks
an assignment \pnt{p} of $Q$ and checks if it satisfies $H$. If so,
\pnt{p} is returned as a satisfying assignment. Otherwise,
\ti{BuildSSA} removes \pnt{p}~\,from $Q$ and picks a clause $C$ of $H$
falsified by \pnt{p}. The assignments of $\Nbhd{p}{p}{C}$ that are not
in $E$ are added to $Q$. After that, \pnt{p} is added to $E$ as an
examined assignment, pair $(\pnt{p},C)$ is added to \Fi and a new
iteration begins. If $Q$ is empty, $E$ is an SSA with center
\sub{p}{init} and AC-mapping \Fi.

%% file: f3ootnote.tex
We use the set of clauses \s{C_1,\dots,C_k} as an alternative
representation of a CNF formula $C_1 \wedge \dots \wedge C_k$.
\vspace{-10pt}

%% file: f5ootnote.tex
In~\cite{ssp}, the notion of ``uncentered'' SSAs was introduced. The
\label{ftn:ssa}
definition of an uncentered SSA is similar to
Definition~\ref{def:ssa}. The only difference is that one requires that for
every $p_i \in P$,  $\nnbhd{p}{i}{C} \subseteq P$ holds instead of
$\NNbhd{p}{p}{i}{C} \subseteq P$.

%% file: b0uild_path.fig.tex
%
%
\setlength{\intextsep}{5pt}
\begin{wrapfigure}{L}{1.4in}
\small
\begin{tabbing}
aaa\=b\=cc\= dd\= \kill
$\mi{BuildPath}(H,\Fi,\sub{p}{init},\vec{s})$\{\\
\tb{\scriptsize{1}}\>  $\mi{Path} := \mi{nil}$ \\
\tb{\scriptsize{2}}\>  $\vec{p}_1 := \sub{p}{init}$\\
\tb{\scriptsize{3}}\>  $i := 1$ \\
\tb{\scriptsize{4}}\>  while ($\vec{p}_i \neq \vec{s}$) \{\\
\tb{\scriptsize{5}}\Tt   $\mi{Path} := \mi{Extend}(\mi{Path},\vec{p}_i)$ \\
\tb{\scriptsize{6}}\Tt   $C := \ac{\ppnt{p}{i}}$  \\
\tb{\scriptsize{7}}\Tt   $v := \mi{FindVar}(C,\vec{p}_i,\vec{s})$ \\
\tb{\scriptsize{8}}\Tt   $\vec{p}_{i+1} := \mi{FlipVar}(\vec{p}_i,v)$ \\
\tb{\scriptsize{9}}\Tt   $i := i+1$  \} \\
\tb{\scriptsize{10}}\>  return($\mi{Path}$) \}\\
\end{tabbing} 
\vspace{-20pt}
\caption{\ti{BuildPath} procedure}
\label{fig:bld_path}
\end{wrapfigure}

%% file: b1uild_ssa.fig.tex
%
%
\begin{wrapfigure}{L}{1.6in}
\small
\vspace{-5pt}
\begin{tabbing}
aaa\=bb\=cc\= dd\= \kill
$\mi{BuildSSA}(H)$\{\\
\tb{\scriptsize{1}}\> $E = \emptyset$;  $\Fi := \emptyset$  \\
\tb{\scriptsize{2}}\> $\sub{p}{init} := \mi{PickInitAssgn}(H)$\\
\tb{\scriptsize{3}}\> $Q := \s{\sub{p}{init}}$   \\
\tb{\scriptsize{4}}\>  while ($Q \neq \emptyset$) \{\\
\tb{\scriptsize{5}}\Tt  $\pnt{p} := \mi{PickAssgn}(\mi{Q})$  \\
\tb{\scriptsize{6}}\Tt  $\mi{Q} := \mi{Q} \setminus \s{\pnt{p}}$ \\
\tb{\scriptsize{7}}\Tt  if $(\mi{SatAssgn}(\pnt{p},H))$ \\
\tb{\scriptsize{8}}\ttt   return($\pnt{p},\mi{nil},\mi{nil},\mi{nil}$) \\
\tb{\scriptsize{9}}\Tt $C := \mi{PickFlsCls}(H,\pnt{p})$   \\
\tb{\scriptsize{10}}\Tt $R:= \Nbhd{p}{p}{C} \setminus E$ \\
\tb{\scriptsize{11}}\Tt $Q := Q \cup R$ \\
\tb{\scriptsize{12}}\Tt $E := E \cup \s{\pnt{p}}$ \\
\tb{\scriptsize{13}}\Tt $\Fi := \Fi \cup \s{(\pnt{p},C)}$\} \\
\tb{\scriptsize{14}}\> return($\mi{nil},E,\sub{p}{init},\Fi$) \}\\
\end{tabbing} 
\vspace{-20pt}
\caption{\ti{BuildSSA} procedure}
\vspace{3pt}
\label{fig:bld_ssa}
\end{wrapfigure}

%% file: c1ts.tex
\section{Complete Test Sets}
\label{sec:cts}
\input{m0iter.fig}
Let $N(X,Y,z)$ be a single-output combinational circuit where $X$ and
$Y$ specify the input and internal variables of $N$ respectively and
$z$ specifies the output variable of $N$. Let $N$ consist of gates
$G_1,\dots,G_k$.  Then $N$ can be represented as $F_N = F_{G_1} \wedge
\dots \wedge F_{G_k}$ where $F_{G_i},i=1,\dots,k$ is a CNF formula
specifying the consistent assignments of gate $G_i$. Proving $N \equiv
0$ reduces to showing that formula $F_N \wedge z$ is unsatisfiable.
\begin{example}
  \label{exmp:circ}
 Circuit $N$ shown in Figure~\ref{fig:miter} represents equivalence
 checking of expressions $(x_1 \vee x_2) \wedge x_3$ and $(x_1 \wedge
 x_3) \vee (x_2 \wedge x_3)$ specified by gates $G_1,G_2$ and
 $G_3,G_4,G_5$ respectively. Formula $F_N$ is equal to $F_{G_1} \wedge
 \dots \wedge F_{G_6}$ where, for instance, $F_{G_1} = C_1 \wedge C_2
 \wedge C_3$, $C_1 = x_1 \vee x_2 \vee \overline{y}_1$, $C_2 =
 \overline{x}_1 \vee y_1$, $C_3 = \overline{x}_2 \vee y_1$.  Every
 satisfying assignment to \V{F_{G_1}} corresponds to a consistent
 assignment to gate $G_1$ and vice versa. For instance,
 $(x_1=0,x_2=0,y_1=0)$ satisfies $F_{G_1}$ and is a consistent
 assignment to $G_1$ since the latter is an OR gate. Formula $F_N
 \wedge z$ is unsatisfiable due to functional equivalence of
 expressions $(x_1 \vee x_2) \wedge x_3$ and $(x_1 \wedge x_3) \vee
 (x_2 \wedge x_3)$. Thus, $N \equiv 0$.
\end{example}

Let \pnt{x} be a test i.e. an assignment to $X$.  The set of
assignments to \V{N} sharing the same assignment \pnt{x} to $X$ forms
a cube of $2^{|Y|+1}$ assignments. $($Recall that $\V{N} = X \cup Y
\cup \s{z}).$ Denote this set as \cube{x}. Only one assignment of
\cube{x} specifies the correct execution trace produced by $N$ under
\pnt{x}.  All other assignments can be viewed as ``erroneous'' traces
under test \pnt{x}.
%

\input{s2em_str.fig}

\begin{definition}
\label{def:cts}
  Let $T$ be a set of tests \s{\ppnt{x}{1},\dots,\ppnt{x}{k}} where $k
  \leq 2^{|X|}$.  We will say that $T$ is a \tb{Complete Test Set
    (CTS)} for $N$ if $\Cube{x}{1} \cup \dots \cup \Cube{x}{k}$
  contains an SSA for formula $F_N \wedge z$.
\end{definition}

If $T$ satisfies Definition~\ref{def:cts}, set $\Cube{x}{1} \cup \dots
\cup \Cube{x}{k}$ ``contains'' a proof that $N \equiv 0$ and so $T$
can be viewed as complete. If $k = 2^{|X|}$, $T$ is the \ti{trivial}
CTS. In this case, $\Cube{x}{1} \cup \dots \cup \Cube{x}{k}$ contains
the trivial SSA consisting of all assignments to \V{F_N \wedge
  z}. Given an SSA $P$ of $F_N \wedge z$, one can easily generate a
CTS by extracting all different assignments to $X$ that are present in
the assignments of $P$.

\begin{example}
  Formula $F_N \wedge z$ of Example~\ref{exmp:circ} has an~SSA of 21
  assignments to \V{F_N\!\wedge\!z}. They have only~5 different
  assignments to\,\,$X\!=\!\s{x_1,\!x_2,\!x_3}$.  The set
  $\{101,\!100,\!011,\!010,\!000\}$ of\,\,those assignments\,is a CTS for $N$.
\end{example}

Definition~\ref{def:cts} is meant for circuits that are not ``too
redundant''.  Highly-redundant circuits are discussed in report
\cite{cmpl_tst} and Appendix~\ref{app:red}.

%% file: m0iter.fig.tex
\setlength{\intextsep}{4pt}
\begin{wrapfigure}{L}{1.7in}
 \begin{center}
    \includegraphics[width=1.6in,height=2.1in]{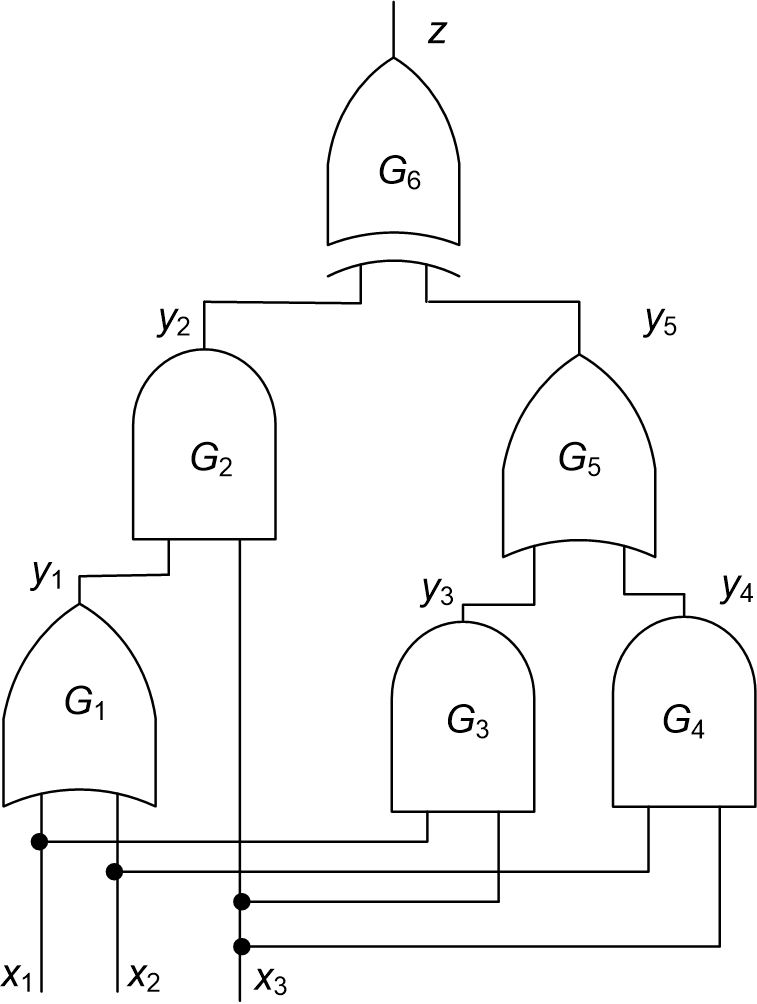}
  \end{center}
\vspace{-10pt}
\caption{Example of circuit $N(X,Y,z)$}
\vspace{5pt}
\label{fig:miter}
\end{wrapfigure}

%% file: s2em_str.fig.tex
%
%
\begin{wrapfigure}{L}{1.6in}
\small
\vspace{-5pt}
\begin{tabbing}
aa\=bb\=cc\= dd\= \kill
$\mi{SemStr}(G,V)$\{\\
\tb{\scriptsize{1}}\> $H := \emptyset$ \\
\tb{\scriptsize{2}}\> while (\ti{true}) \{ \\
\tb{\scriptsize{3}}\Tt $(\pnt{v},\!P)\!:=\!\mi{BuildSSA}(H)$\\
\tb{\scriptsize{4}}\Tt if ($P \neq \mi{nil}$)  \\
\tb{\scriptsize{5}}\ttt return($\mi{nil},P$)\\
\tb{\scriptsize{6}}\Tt  $(C,\pnt{s}) := \mi{GenCls}(G,V,\pnt{v})$ \\
\tb{\scriptsize{7}}\Tt  if ($\pnt{s} \neq \mi{nil}$) \\
\tb{\scriptsize{8}}\ttt return($\pnt{s},\mi{nil}$)\\
\tb{\scriptsize{9}}\Tt  $H := H \cup \s{C}$ \} \\
\end{tabbing} 
\vspace{-20pt}
\caption{\ti{SemStr} procedure}
\label{fig:sem_str}
\end{wrapfigure}

%% file: a1lgor.tex
\section{\sas Procedure}
\label{sec:algor}
\subsection{Motivation}

  Building an SSA for a large formula is inefficient. So, constructing
  a CTS of $N$ from an SSA of $F_N \wedge z$ is impractical. To
  address this problem, we introduce a procedure called \sas (a short
  for ``Semantics and Structure''). Given formula $F_N \wedge z$ and a
  set of variables $V \subseteq \V{F_N \wedge z}$, \sas generates a
  simpler formula $H(V)$ implied by $F_N \wedge z$ at the same time
  trying to build an SSA for $H$. If \sas succeeds in constructing
  such an SSA, formula $H$ is unsatisfiable and so is $F_N \wedge z$.
  Then a set of tests $T$ is extracted from this SSA. As we show in
  Subsection~\ref{ssec:approx}, one can view $T$ as an approximation
  of a CTS for $N$ (if $X \subseteq V$) or an ``approximation of
  approximation'' of a CTS (if $X \not\subseteq V$).

  \begin{example}
  Consider the circuit $N$ of Figure~\ref{fig:miter} where
  $X=\s{x_1,x_2,x_3}$.  Assume that $V = X$. Application of \sas to
  $F_N \wedge z$ produces $H(X)= (\overline{x}_1 \vee \overline{x}_3)
  \wedge (\overline{x}_2 \vee \overline{x}_3) \wedge (x_1 \vee x_2)
  \wedge x_3$. \sas also generates an SSA of $H$ of four assignments
  to $X$: \s{000,001,011,101} with center \sub{p}{init}=000. (We omit
  the AC-mapping here.) These assignments form an approximation of CTS
  for $N$.
\end{example}
  \subsection{Description of \sas}
The pseudocode of \sas is shown in Figure~\ref{fig:sem_str}.  \sas
accepts formula $G$ (in our case, $G := F_N \wedge z$) and a set of
variables $V \subseteq \V{G}$. \sas outputs an assignment satisfying
$G$ or formula $H(V)$ implied by $G$ and an SSA of $H$.  Originally,
the set of clauses $H$ is empty. $H$ is computed in a \ti{while}
loop. First, \sas tries to build an SSA for the current formula $H$ by
calling \ti{BuildSSA} (line 3).  If $H$ is unsatisfiable,
\ti{BuildSSA} computes an SSA $P$ returned by \sas (line
5). Otherwise, \ti{BuildSSA} returns an assignment \pnt{v} satisfying
$H$. In this case, \sas calls procedure \ti{GenCls} to build a clause
$C$ falsified by \pnt{v}.  Clause $C$ is obtained by resolving clauses
of $G$ on variables of $W$. (Hence $C$ is implied by $G$.)  If \pnt{v}
can be extended to an assignment \pnt{s} satisfying $G$, \sas
terminates (lines 7-8). Otherwise, $C$ is added to $H$ and a new
iteration begins.

\input{g3en_clause.fig}

Procedure \ti{GenCls} is shown in Figure~\ref{fig:gen_cls}. First,
\ti{GenCls} generates formula \cof{G}{v} obtained from $G$ by
discarding clauses satisfied by \pnt{v} and removing literals
falsified by \pnt{v}. Then \ti{GenCls} checks if there is an
assignment \pnt{s} satisfying \cof{G}{v}.  If so, $\pnt{s} \cup
\pnt{v}$ is returned as an assignment satisfying $G$. Otherwise, a
proof $R$ of unsatisfiability of \cof{G}{v} is produced.  Then
\ti{GenCls} forms a set $V' \subseteq V$. A variable $w$ is in $V'$
iff a clause of \cof{G}{v} is used in proof $R$ and its parent clause
from $G$ has a literal of $w$ falsified by \pnt{v}. Finally, clause
$C$ is generated as a disjunction of literals of $V'$ falsified by
\pnt{v}. By construction, clause $C$ is implied by $G$ and falsified
by \pnt{v}.

%

%% file: g3en_clause.fig.tex
%
%
\begin{wrapfigure}{L}{1.2in}
\small
\begin{tabbing}
aa\=bb\=cc\= dd\= \kill
$\mi{GenCls}(G,V,\pnt{v})$\{\\
\tb{\scriptsize{1}}\> $\cof{G}{v} := \mi{GenForm}(F,\pnt{v})$ \\
\tb{\scriptsize{2}}\> $(\pnt{s},R) := \mi{ChkSat}(\cof{G}{v})$\\
\tb{\scriptsize{3}}\> if ($\pnt{s} \neq \mi{nil}$) \\
\tb{\scriptsize{4}}\Tt return($\mi{nil},\pnt{s} \cup \pnt{v}$)\\
\tb{\scriptsize{5}}\> $V' := \mi{Analyze}(R,\cof{G}{v},G)$\\
\tb{\scriptsize{6}}\> $C := \mi{FormCls}(V',\pnt{v})$ \\
\tb{\scriptsize{7}}\>  return($C,\mi{nil}$)\\
\end{tabbing} 
\vspace{-15pt}
\caption{\ti{GenCls} procedure}
\label{fig:gen_cls}
\end{wrapfigure}

%% file: i2nternal_cut.tex
\section{Building Approximations Of CTS}
\label{sec:app_cts}
%
%
\subsection{Two kinds of approximations of CTSs}
\label{ssec:approx}
\input{g4en_tests.fig}

%
As before, let $H(V)$ denote a formula implied by $F_N \wedge z$ that
is generated by \sas and $P$ denote an SSA for $H$.  Projections of
$N$ can be of two kinds depending on whether $X \subseteq V$
holds. Let $X \subseteq V$ hold and $T$ be the test set extracted from
$P$ as described in Section~\ref{sec:cts}. That is $T$ consists of all
different assignments to $X$ present in the assignments of $P$.  On
one hand, using the reasoning of Section~\ref{sec:cts} one can show
that $T$ is a CTS for projection of $N$ on $V$. On the other hand,
since $H(V)$ is essentially an abstraction of $F_N \wedge z$, set $T$
is an approximation of a CTS for $N$. For that reason, we will refer
to $T$ as a \tb{CTS\textsuperscript{a}} of $N$ where superscript ``a''
stands for ``approximation''.

Now assume that $X \not\subseteq V$ holds. Generation of a test set
$T$ from $P$ for this case is described in the next section. The set
$T$ can be viewed as an approximation of a set $T'$ built for
projection of $N$ on set $V \cup X$. Since $T'$ is a \acts for $N$, we
will refer to $T$ as \tb{CTS\textsuperscript{aa}} where ``aa'' stands
for ``approximation of approximation''.

%
%
\subsection{Construction of \aacts}

Consider extraction of a test set $T$ from SSA $P$ of formula $H(V)$
when $X \not\subseteq V$. Since $V$, in general, contains internal
variables\footnote{If the special case $V \subset X$ holds, every
  assignment of $P$ can be easily turned into a test by assigning
  values to variables of $X \setminus V$ (e.g. randomly).}
of $N$, translation of $P$ to a test set $T$ needs a special procedure
\ti{GenTests} shown in Figure~\ref{fig:gen_tests}. For every
assignment \pnt{v} of $P$, \ti{GenTests} checks if formula $F_N$ is
satisfiable under assignment \pnt{v} (i.e. if there exists a test
under which $N$ assigns \pnt{v} to $V$). If so, an assignment \pnt{x}
to $X$ is extracted from the satisfying assignment and added to $T$ as
a test. Otherwise, \ti{GenTests} runs a \ti{for} loop (lines 8-13) of
$\mi{Tries}$ iterations. In every iteration, \ti{GenTests} relaxes
$F_N$ by removing the clauses specifying a small subset of gates
picked randomly.  If the relaxed version of $F_N$ is satisfiable, a
test is extracted from the satisfying assignment and added to $T$.
%
%
\subsection{Finding a set of variables to project on}
\label{ssec:int_cut}
%
\input{i3nt_cut.fig}

%
Intuitively, a good choice of the set $V$ to project $N$ on is a
(small) coherent subset of variables of $N$ reflecting its structure
and/or semantics. One obvious choice of $V$ is the set $X$ of input
variables of $N$.  In this section, we describe generation of a set
$V$ whose variables form an internal cut of $N$ denoted as \ti{Cut}.
Procedure \ti{GenCut} for generation of set \ti{Cut} consisting of
\ti{Size} gates is shown in Figure~\ref{fig:int_cut}. Set $V$ is
formed from output variables of the cut gates.

The current cut is specified by $\mi{Gts} \cup \mi{Inps}$.  Set
\ti{Gts} is initialized with the output gate \Sub{G}{out} of circuit
$N$ and \ti{Inps} is originally empty. \ti{GenCut} computes the
\ti{depth} of every gate of \ti{Gts}. The depth of \Sub{G}{out} is set
to 0. Set \ti{Gts} is processed in a \ti{while} loop (lines 5-15). In
every iteration, a gate of the smallest depth is picked from
\ti{Gts}. Then \ti{GenCut} removes gate $G$ from \ti{Gts} and examines
the fan-in gates of $G$ (lines 9-15). Let $G'$ be a fan-in gate of $G$
that has not been seen yet and is not a primary input of $N$. Then the
depth of $G'$ is set to that of $G$ plus 1 and $G'$ is added to
\ti{Gts}. If $G'$ is a primary input of $N$ it is added to \ti{Inps}.

%% file: g4en_tests.fig.tex
%
%
\begin{wrapfigure}{L}{1.5in}
\small
\begin{tabbing}
aa\=bb\=cc\= dd\= \kill
$\mi{GenTests}(F_N,X,P,\mi{Tries})$\{\\
\tb{\scriptsize{1}}\>  $\mi{T} := \emptyset$ \\
\tb{\scriptsize{2}}\> for each $\pnt{v} \in P$ \{\\
\tb{\scriptsize{3}}\Tt  $\pnt{s} := \mi{SatAssgn}(F_N,\pnt{v})$\\
\tb{\scriptsize{4}}\Tt  if ($\pnt{s} \neq \mi{nil}$) \{\\
\tb{\scriptsize{5}}\ttt   $\pnt{x} := \mi{ExtrTest}(\pnt{s},X)$ \\
\tb{\scriptsize{6}}\ttt  $T:= T \cup \pnt{x}$\} \\
\tb{\scriptsize{7}}\Tt  else \\
\tb{\scriptsize{8}}\ttt for $(i=0;\!i < \mi{Tries};\!i\verb!++!)$\!\{ \\
\tb{\scriptsize{9}}\tttt   $F^*_N := \mi{Relax}(F_N)$ \\
\tb{\scriptsize{10}}\tttt   $\pnt{s} := \mi{SatAssgn}(F^*_N,\pnt{v})$ \\
\tb{\scriptsize{11}}\tttt   if ($\pnt{s} = \mi{nil}$) continue\\
\tb{\scriptsize{12}}\tttt  $\pnt{x} := \mi{ExtrTest}(\pnt{s},X)$\} \\
\tb{\scriptsize{13}}\tttt  $T:= T \cup \pnt{x}$\}\} \\
\tb{\scriptsize{14}}\>~ return($T$)\}\\
\end{tabbing} 
\vspace{-20pt}
\caption{\ti{GenTests} procedure}
\label{fig:gen_tests}
\end{wrapfigure}

%% file: i3nt_cut.fig.tex
%
%
\begin{wrapfigure}{l}{1.5in}
\small
\begin{tabbing}
aa\=bb\=cc\=dd\= \kill
$\mi{GenCut}(N,\mi{Size})$\{\\
\tb{\scriptsize{1}}\> $\Sub{G}{out} := \mi{OutGate}(N)$ \\
\tb{\scriptsize{2}}\>$\mi{Gts}:= \s{\Sub{G}{out}}$ \\
\tb{\scriptsize{3}}\> $Dpth(\Sub{G}{out}):= 0$ \\
\tb{\scriptsize{4}}\> $\mi{Inps}:= \emptyset$ \\
\tb{\scriptsize{5}}\> while $(|\mi{Gts}\!\cup\!\mi{Inps}|\!<\!\mi{Size})$ \{\\
\tb{\scriptsize{6}}\Tt  $G\!:=\!\mi{MinDepth}(\mi{Gts},\!\mi{Dpth})$\\
\tb{\scriptsize{7}}\Tt  $\mi{Gts} := \mi{Gts} \setminus \s{G}$ \\
\tb{\scriptsize{8}}\Tt  $\mi{Seen}(G) := \mi{true}$ \\
\tb{\scriptsize{9}}\Tt foreach $G' \in \mi{FanIn}(G) \{$\\
\tb{\scriptsize{10}}\ttt  if $(\mi{Seen}(G'))$ continue \\
\tb{\scriptsize{11}}\ttt  if $(G' \in \mi{Inputs}(N))$ \{ \\
\tb{\scriptsize{12}}\tttt   $\mi{Inps} = \mi{Inps} \cup \s{G'}$ \\
\tb{\scriptsize{13}}\tttt   continue \}\\
\tb{\scriptsize{14}}\ttt $\mi{Dpth}(G')\!:=\!\mi{Dpth}(G)\!+\!1$ \\
\tb{\scriptsize{15}}\ttt $\mi{Gts} := \mi{Gts} \cup \s{G'}$\}\} \\
\tb{\scriptsize{16}}\>~return($\mi{Gts} \cup \mi{Inps}$)\} \\
\end{tabbing} 
\vspace{-20pt}
\caption{\ti{GenCut} procedure}
\label{fig:int_cut}
\end{wrapfigure}

%% file: a2pplications.tex
\section{Applications Of Property-Checking Tests}
\label{sec:appl}
Given a multi-output circuit $M$, traditional testing is used to
verify $M$ ``as a whole''. In this paper, we describe generation of a
test set meant for checking a \ti{particular property} of $M$
specified by a single-output circuit $N$. In this section, we present
some applications of property-checking test sets.

\subsection{Testing properties specified by similar circuits}
\label{ssec:des_changes}

Let $N$ be a single-output circuit and $T$ be a test set generated
when proving $N\equiv0$. Let $N^*$ be a circuit that is similar to
$N$.  (For instance, $N$ can specify a property of a circuit $M$
whereas $N^*$ specifies the same property after a modification of
$M$.)  Then one can use $T$ to verify if $N^* \equiv 0$.  Since $T$ is
generated for a similar circuit $N$, there is a good chance that it
contains a counterexample to $N^* \equiv 0$, if any. (Of course, the
fact that $N^*$ evaluates to 0 for all tests of $T$ does not mean that
$N^* \equiv 0$ even if $T$ is a CTS for $N$). In
Subsection~\ref{ssec:bug}, we give experimental evidence supporting
the observation above.

Assuming that $N \equiv 0$ was proved formally, checking if $N^*
\equiv 0$ holds can be verified formally too. So applying tests of $T$
to $N^*$ can be viewed as a ``light'' verification procedure for
exposing bugs. On the other hand, one can re-use test $T$ in
situations where the necessity to apply a formal tool is overlooked or
formal methods are not powerful enough. Let $N$ specify a property
$\xi$ of a \ti{component} of a design $D$. Suppose that this component
is modified under assumption that preserving $\xi$ is not necessary
any more.  By applying $T$ to $D$ one can invoke behaviors that break
$\xi$ and expose a bug in $D$, if any, caused by ignoring $\xi$. If
$D$ is a large design, finding such a bug by formal verification may
not be possible.

%
%
\subsection{Verification of corner cases}
\label{ssec:corners}

\input{s3ubcirc.fig}

Let $K$ be a single-output subcircuit of circuit $M$ as shown in
Figure~\ref{fig:subcirc}. For the sake of simplicity we consider here
the case where the set $X_K$ of input variables of $K$ is a subset of
the set $X$ of input variables of $M$. (The technique below can also
be applied when input variables of $K$ are \ti{internal} variables of
$M$.)  Suppose $K$ evaluates, say, to value 0 much more frequently then
to 1. Then one can view an input assignment of $M$ for which $K$
evaluates to 1 as specifying a ``corner case'' i.e. a rare
event. Hitting such a corner case by a random test can be very hard.
This issue can be addressed by using a coverage metric that
\ti{requires} setting the value of $K$ to both 0 and 1.  (The task of
finding a test for which $K$ evaluates to 1 can be solved, for
instance, by a SAT-solver.)  The problem however is that hitting a
corner case only once may be insufficient.

One can increase the frequency of hitting the corner case above as
follows.  Let $N$ be a miter of circuits $K'$ and $K''$ (see
Figure~\ref{fig:gen_miter}) i.e.  a circuit that evaluates to 1 iff
$K'$ and $K''$ are functionally inequivalent.  Let $K'$ and $K''$ be
two copies of circuit $K$. So $N \equiv 0$ holds. Let test set $T_K$
be extracted from an SSA built for a projection of $N$ on a set $V
\subseteq \V{N}$. Set $T_K$ can be viewed as a result of ``squeezing''
the truth table of $K$. Since this truth table is dominated by input
assignments for which $K$ evaluates to 0, this part of the truth table
is \ti{reduced the most}. So, one can expect that the ratio of tests
of $T_K$ for which $K$ evaluates to 1 is higher than in the truth
table of $K$. In Subsection~\ref{ssec:ecorners}, we substantiate this
intuition experimentally. One can easily extend an assignment
\ppnt{x}{K} of $T_K$ to an assignment \pnt{x} to $X$ e.g. by randomly
assigning values to the variables of $X \setminus X_K$.

\subsection{Dealing with incomplete specifications}
\label{ssec:incomp_spec}
One can use property-checking tests to mitigate the problem of
incomplete specifications.  By running tests generated for an
incomplete set of properties of $M$, one can expose bugs overlooked
due to missing some properties. An important special case of this
problem is as follows. Let $\xi$ be a property of $M$ that
holds. Assume that the correctness of $M$ requires proving a slightly
\ti{different} property $\xi'$ that is not true. By running a test set
$T$ built for property $\xi$, one may expose a bug overlooked in
formal verification due to proving $\xi$ instead of $\xi'$.  In
Subsection~\ref{ssec:missed_props}, we illustrate the idea above
experimentally.

\subsection{Testing sequential circuits}
\label{ssec:seq_circ}
There are a few ways to apply property-checking tests meant for
combinational circuits to verification of \ti{sequential} circuits.
Here is one of them based on bounded model checking~\cite{bmc}. Let
$M$ be a sequential circuit and $\xi$ be a property of $M$. Let
$N(X,Y,z)$ be a circuit such that $N \equiv 0$ holds iff $\xi$ is true
for $k$ time frames. Circuit $N$ is obtained by unrolling $M$ $k$
times and adding logic specifying property $\xi$. Set $X$ consists of
the subset $X'$ specifying the state variables of $M$ in the first
time frame and subset $X''$ specifying the combinational input
variables of $M$ in $k$ time frames.

\input{g1en_miter.fig}

Having constructed $N$, one can build CTSs, \acts{s} and \aacts{s} for
testing property $\xi$ of $M$. The only difference here from the
problem we have considered so far is as follows. Circuit $M$ starts in
a state satisfying some formula $I(X')$ that specifies the initial
states. So, one needs to check if $N \equiv 0$ holds only for the
assignments to $X$ satisfying $I(X')$. A test here is an assignment
$(\ppnt{x'}{1},\ppnt{x''}{1},\dots,\ppnt{x''}{k})$ where \ppnt{x'}{1}
is an initial state and \ppnt{x''}{i}, $1 \leq i \leq k$ is an
assignment to the combinational input variables of $i$-th time frame.
Given a test, one can easily compute the corresponding sequence of
states $(\ppnt{x'}{1},\dots,\ppnt{x'}{k})$ of $M$.  In
Subsection~\ref{ssec:missed_props}, we give an example of building an
\aacts for a sequential circuit.

%% file: s3ubcirc.fig.tex
\setlength{\intextsep}{4pt}
\begin{wrapfigure}{L}{1.15in}
 \begin{center}
    \includegraphics[width=1in]{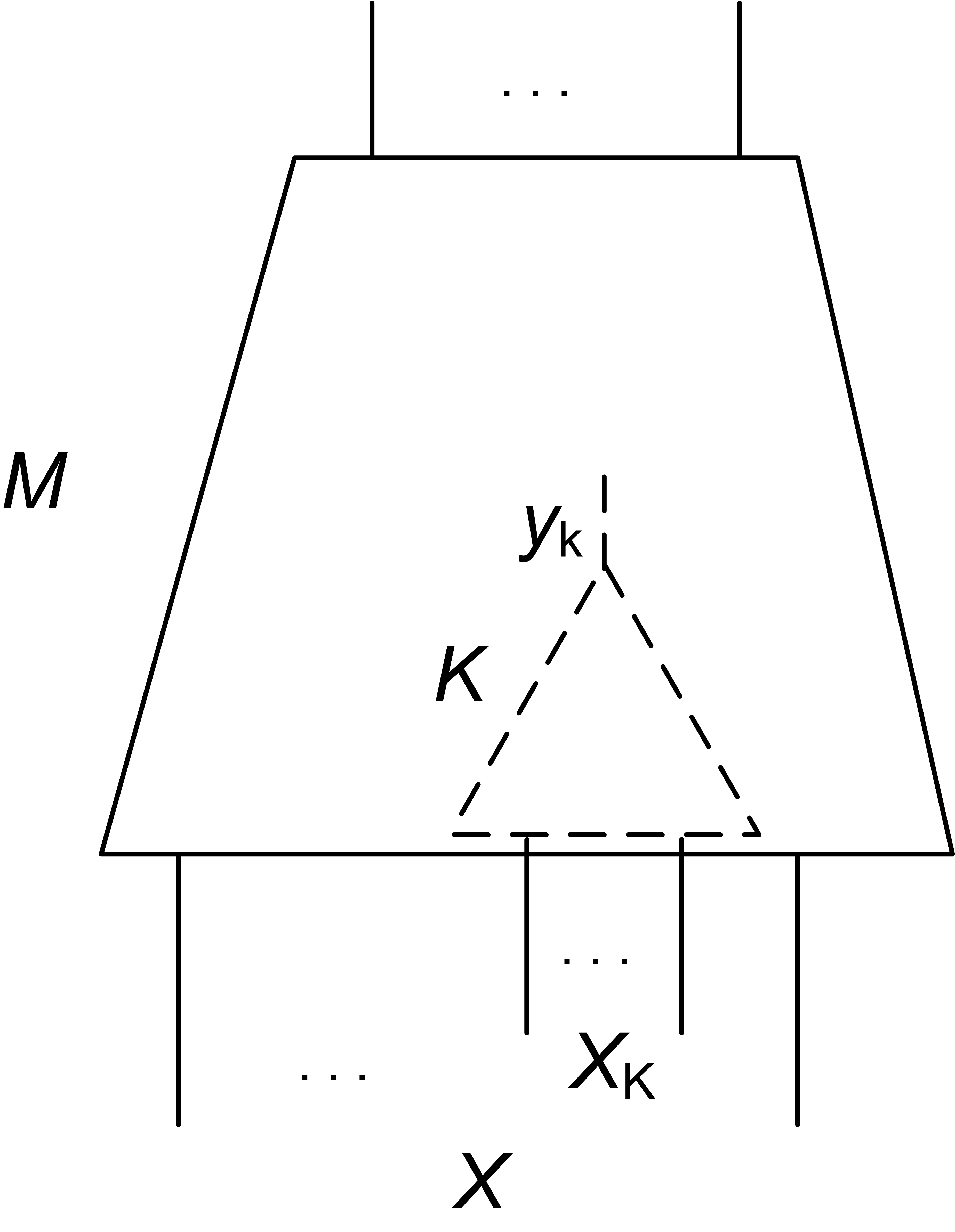}
  \end{center}
\vspace{-10pt}
\caption{Subcircuit $K$ of circuit $M$}
\vspace{10pt}
\label{fig:subcirc}
\end{wrapfigure}

%% file: g1en_miter.fig.tex
\begin{wrapfigure}{L}{1.45in}
 \begin{center}
    \includegraphics[width=1.3in]{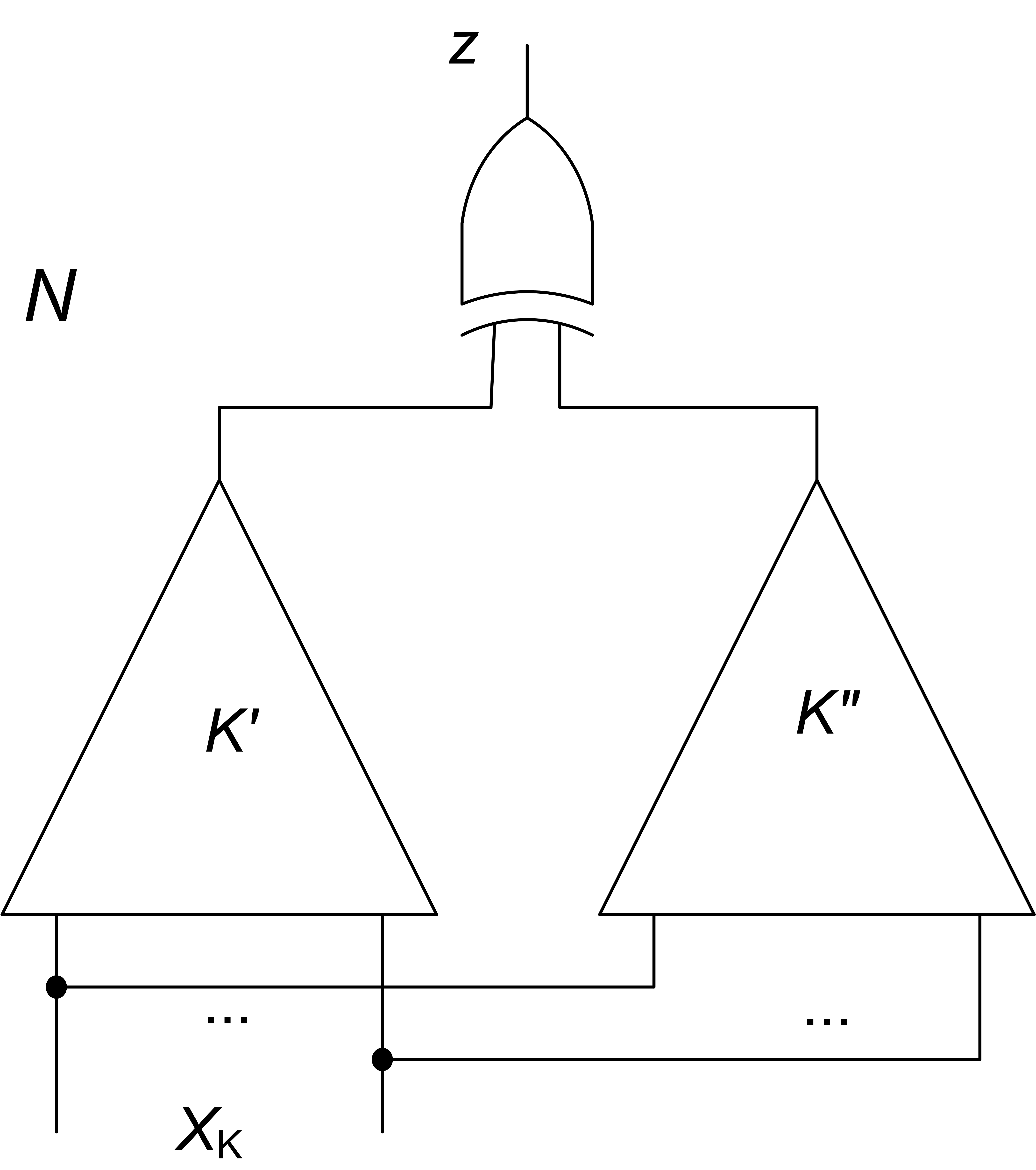}
  \end{center}
\vspace{-10pt}
\caption{The miter of circuits $K'$ and $K''$}
\label{fig:gen_miter}
\end{wrapfigure}

%% file: e4xper.tex
\section{Experiments}
\label{sec:exper}
\input{g5en_pc_tests}

%
In this section, we describe experiments with property-checking tests
(PCT) generated by procedure \ti{GenPCT} shown in
Figure~\ref{fig:gen_pct}. \ti{GenPCT} accepts a single-output circuit
$N$ and outputs a set of tests $T$. (For the sake of simplicity, we
assume here that $N \equiv 0$ holds.) \ti{GenPCT} starts with
generating formula $F_N \wedge z$ and a set of variables $V \subseteq
\V{F_N \wedge z}$. Then it calls \sas (see Fig.~\ref{fig:sem_str}) to
compute an SSA $P$ of formula $H(V)$ describing a projection of
circuit $N$ on $V$\!. If $H(V)$ does not depend on a variable $w \in V$,
all assignments of $P$ have the same value of $w$. Procedure
\ti{Diversify} randomizes the value of $w$ in the assignments of $P$.
Finally, \ti{BldTests} uses $P$ to extract a test set for circuit
$N$. If $X \subseteq V$ holds (where $X$ is the set of input variables
of $N$), \ti{BldTests} outputs all the different assignments to $X$
present in assignments of $P$. Otherwise, \ti{BldTests} calls
procedure \ti{GenTests} (see Fig.~\ref{fig:gen_tests}).

If $V = \V{F_N \wedge z}$, then $H(V)$ is $F_N \wedge z$ itself and
\ti{GenPCT} produces a CTS of $N$. Otherwise, according to definitions
of Subsection~\ref{ssec:approx}, \ti{GenPCT} generates a \acts (if $X
\subseteq V$) or \aacts (if $X \not\subseteq V$).

In the following subsections, we describe results of four experiments.
In the first three experiments we used circuits specifying next state
functions of latches of HWMCC-10 benchmarks. (The motivation was to
use realistic circuits.)  In our implementation of \sas, as a
SAT-solver, we used Minisat 2.0~\cite{minisat,minisat2.0}.  We also
employed Minisat to run simulation. To compute the output value of $N$
under test \pnt{x}, we added unit clauses specifying \pnt{x} to
formula $F_N \wedge z$ and checked its satisfiability.

\subsection{Comparing  CTSs, \acts{s} and \aacts{s}}
\label{ssec:cts}
%
\input{c5ts.tbl}

The objective of the first experiment was to give examples of circuits
with non-trivial CTSs and compare the efficiency of computing CTSs,
\acts{s} and \aacts{s}.  In this experiment, $N$ was a miter
specifying equivalence checking of circuits $M'$ and $M''$ (see
Figure~\ref{fig:gen_miter}).  $M''$ was obtained from $M'$ by
optimizing the latter with ABC~\cite{abc}.

The results of the first experiment are shown in Table~\ref{tbl:cts}.
The first two columns specify an HWMCC-10 benchmark and its latch
whose next state function was used as $M'$.  The next two columns give
the number of input variables and that of gates in the miter $N$. The
following pair of columns describe computing a CTS for $N$.  The first
column of this pair gives the size of the SSA $P$ found by \ti{GenPCT}
in thousands.  The number of tests in the set $T$ extracted from $P$
is shown in the parentheses in thousands. The second column of this
pair gives the run time of \ti{GenPCT} in seconds.

The last four columns of Table~\ref{tbl:cts} describe results of
computing test sets for a projection of $N$ on a set of variables
$V$. The first column of this group shows if \acts or \aacts was
computed whereas the next column gives the size of $V$. The third
column of this group provides the size of SSA $P$ and the test set $T$
extracted from $P$ (in parentheses). Both sizes are given in
thousands. The last column shows the run time of \ti{GenPCT}. For the
first five examples, we used a projection of $N$ on $X$, thus
constructing a \acts of $N$. For the last four examples we computed a
projection of $N$ on an internal cut (see
Subsection~\ref{ssec:int_cut}) thus generating a \aacts of $N$.  \pct
was called with parameter $\mi{Tries}$ set to 5 (see
Fig.~\ref{fig:gen_tests} and~\ref{fig:gen_pct}).

For the first three examples, \pct managed to build non-trivial CTSs
that are smaller than $2^{|X|}$.  For instance, the trivial CTS for
example \ti{bob3} consists of $2^{14}$=16,384 tests, whereas \pct
found a CTS of 2,004 tests. (So, to prove $M'$ and $M''$ equivalent it
suffices to run 2,004 out of 16,384 tests.) For the other examples,
\pct failed to build a non-trivial CTS due to exceeding the memory
limit (1.5 Gbytes). On the other hand, \pct built a \acts or \aacts
for all nine examples of Table~\ref{tbl:cts}. Note, however, that
\acts{s} give only a moderate improvement over CTSs. For the last four
examples \pct failed to compute an \acts of $N$ due to memory overflow
whereas it had no problem computing an \aacts of $N$. So \aacts{s} can
be computed efficiently even for large circuits. Further, we show that
\aacts{s} are also very effective.

\input{e7xper_bug_hunting}

\input{e8xper_corner_cases}

\input{e9xper_missed_props}

%% file: g5en_pc_tests.tex
%
\begin{wrapfigure}{L}{1.4in}
\small
\begin{tabbing}
aa\=bb\=cc\= dd\= \kill
$\mi{GenPCT}(N,\mi{Tries})$\{\\
\tb{\scriptsize{1}}\> $F_N\wedge z:= \mi{GenForm}(N)$\\
\tb{\scriptsize{2}}\> $V := \mi{GenVars}(F_N \wedge z)$\\
\tb{\scriptsize{3}}\> $P := \mi{SemStr}(F_N \wedge z, V)$\\
\tb{\scriptsize{4}}\> $P := \mi{Diversify}(P)$\\
\tb{\scriptsize{5}}\> $T := \mi{BldTests}(F_N,P,\mi{Tries})$ \\
\tb{\scriptsize{6}}\> return($T$)\}\\
\end{tabbing} 
\vspace{-20pt}
\caption{\ti{GenPCT} procedure}
\label{fig:gen_pct}
\end{wrapfigure}

%% file: c5ts.tbl.tex
%
%
\begin{table}
\small
\caption{\ti{Computing CTSs, \acts{s} and \aacts{s}}}
\vspace{-10pt}
\scriptsize
\begin{center}
\begin{tabular}{|@{}l@{}|@{}l@{}|@{}l@{}|@{}l@{}|@{}l@{}|@{}l@{}|@{}l@{\,}|@{\,}l@{}|@{\,}l@{~}|@{~}l@{~}|} \hline
~name     & \,la- &\,\#inp\,&\,\#ga-   & \multicolumn{2}{c|}{CTS  } & \multicolumn{4}{c|}{\acts or \aacts }  \\ \cline{5-10}
         &\,tch  &\,vars\,&\,tes          &\,$|\mi{SSA}|$ &\,time\,  &\,test\,&  &$|\mi{SSA}|$&  time  \\ 
  &       &         &           &\,(\#tests)&~(s.) &\,set\, & \,$|V|$\,  &(\#tests)   & (s.) \\
   &       &         &           &~$\times 10^3$ &   &\,type\,       &       & $\times 10^3$  &        \\ \hline
\,bob3 &  \,L26  & \,14    &\,41     &\,46~(2.0)    &\,0.1         &\,\tiny{\acts}  &\,14  &\,0.6~(0.6)\,& 0.01  \\ \hline
\,eijks258   &  \,L10  & \,16    &\,45     &\,259~(8.2)&\,0.5    &\,\tiny{\acts}  &\,16  &\,0.1~(0.1)\, & 0.02 \\ \hline
\,cmudme1    &  \,L230 &  \,19    & \,50     &\,2,184~(63)\, &\,5.4     &\,\tiny{\acts} &\,19  &\,13~(13)\,& 0.1 \\ \hline
\,mutexp0    &\,L60  & \,29    & \,199    &  ~memout   & ~$*$    &\,\tiny{\acts}&\,29  &\,659~(659)\,&\,26 \\ \hline
\,pdtpmsmiim\, &\,L118  & \,31    & \,136    &  ~memout  & ~$*$   & \,\tiny{\acts} &\,31  &\,936~(936)\,& 4.2 \\ \hline
\,abp4pold   &\,L270  &\,129    & \,1,178\,&  ~memout  & ~$*$   &\,\tiny{\aacts}\,&\,22   &\,0.9~(0.5)\, & 0.6  \\ \hline
\,pj2009  &\,L1318\,&\,366     & \,25,160\,& ~memout   & ~$*$   & \,\tiny{\aacts}\,   &\,22 &\,0.6~(0.3)\, & 51 \\ \hline
\,mentorb..00  &\,L8670\,&\,626 &\,3,156\, & ~memout    & ~$*$   &\,\tiny{\aacts}\,    &\,22 &\,1.2~(0.6)\, & 11 \\ \hline
\,139454p0&\,L1676\,&\,791     &\,19,843\,& ~memout  & ~$*$   & \,\tiny{\aacts}\,  &\,22  &\,0.1~(0.1)\,  & 99 \\ \hline
\end{tabular}                
\end{center}
\vspace{-15pt}
\label{tbl:cts}
\end{table}

%% file: e7xper_bug_hunting.tex
\subsection{Re-using property-checking tests to detect bugs}
\label{ssec:bug}
In the second experiment, we employed the idea of re-using
property-checking tests (see Subsection~\ref{ssec:des_changes}) to
verify \ti{relaxed} equivalence~\cite{rlx_ec}.  Let circuit \spi{M} be
obtained from circuit $M$ by applying a set of changes $\pi$. Regular
equivalence of $M$ and \spi{M} means that these circuits produce the
same output assignment for the same input assignment. Relaxed
equivalence requires only that the difference between output
assignments of $M$ and \spi{M} is in a \ti{specified range}. So
regular equivalence implies relaxed one and hence the latter is a
\ti{weaker} property than the former. Intuitively, this makes relaxed
equivalence harder for testing (because the space of buggy behaviors
is smaller).

In this experiment, we compared two-output circuits $M$ and $M^{\pi}$.
Namely we checked property $\xi(M,M^{\pi})$ that $(y_1 \equiv
y^{\pi}_1) \vee (y_2 \equiv y^{\pi}_2)$ holds where $y_1,y_2$ and
$y^{\pi}_1,y^{\pi}_2$ specify the outputs of $M$ and $M^{\pi}$
respectively. Property \pim states that the Hamming distance between
the output assignments produced by $M$ and \spi{M} for the same input
assignment is less or equal to 1.  Circuit $M$ was extracted from the
transition relation of an HWMCC-10 benchmark.  Circuit $M^{\pi}$ was
obtained from $M$ by making changes $\pi$ that broke property \pim.

Let \spi{N} denote a circuit specifying property \pim. In the
experiment, we tested \spi{N} using three approaches. The first
approach was to apply tests generated to detect stuck-at faults
(\tb{SAF})~\cite{abram} in $M$. We used SAF tests as an example of a
test set driven by a coverage metric\footnote{Note that SAF tests are
  stronger than tests satisfying traditional coverage metrics
  e.g. those used in software.  In addition to exciting an event in
  $M$, a SAF test must \ti{propagate} the effect of this event to a
  primary output of $M$.}  The second approach was random testing of
\spi{N}. In the third approach we did the following. First, we built a
\aacts for circuit $N^{\emptyset}$ specifying property \pim for the
case where $\pi=\emptyset$ and hence $M^{\pi}$ was identical to
$M$. (Obviously, $N^{\emptyset} \equiv 0$ holds.)  Then we
\ti{re-used} this \aacts to verify circuit \spi{N} for the case $\pi
\neq \emptyset$.

%
\input{f4lt_tst.tbl}

A representative subset of examples we tried is shown in
Table~\ref{tbl:flt_tst}.  The first column gives the name of the
HWMCC-10 benchmark from which circuit $M$ was extracted. The next two
columns show the size of the circuit. The following two columns list
the number of SAF tests and the total run time (i.e. time taken to
generate tests\footnote{\input{f10ootnote}} and run simulation by
Minisat). The next two columns show the performance of random testing.
The last three columns describe testing by \aacts{s}. The first column
of the three gives the size of the SSA for formula $H(V)$ generated by
\pct.  The set of variables $V$, $|V|=18$, forming an internal cut was
generated by procedure \ti{GenCut} (see Fig.~\ref{fig:int_cut}).  The
next column shows the size of a \aacts obtained with $\mi{Tries}$ set
to 100.  The last column gives the total run time.

The test sets with a counterexample breaking \pim are marked with an
asterisk. Table~\ref{tbl:flt_tst} shows that SAF tests failed to
detect a bug, random tests found a bug for two examples and \aacts{s}
succeeded for all examples. (On the other hand, the same SAF tests and
\aacts{s} found bugs in all eight examples modified to check
\ti{regular} equivalence i.e the property $(y_1 \equiv y^{\pi}_1)
\wedge (y_2 \equiv y^{\pi}_2)$. Random testing limited to 100 million
tests found a bug in five examples.) So \aacts{s} proved effective
even in testing a ``weak'' property.

%% file: f4lt_tst.tbl.tex
%
%
\begin{table}[h!]
\small
\caption{\ti{Bug detection. \aacts{s} are computed for $|V|=18$. The number of random
tests is limited to $10^8$. An asterisk marks test sets with a counterexample}}
\vspace{-5pt}
\scriptsize
\begin{center}
  \begin{tabular}{|@{}l@{}|@{}l@{}|@{}l@{}|@{}l@{}|@{}l@{}|@{}l@{}|@{}l@{}|@{}l@{}|@{}l@{}|@{}l@{}|} \hline
~name      &\,\#inp\,&\,\#ga-   & \multicolumn{2}{c|}{SAF} &  \multicolumn{2}{c|}{random} & \multicolumn{3}{c|}{testing by }  \\
           &\,vars\,&\,tes      & \multicolumn{2}{c|}{tests} & \multicolumn{2}{c|}{tests} &\multicolumn{3}{c|}{\aacts}  \\ \cline{4-10}
         &          &    &\,\#tests\,& \,time\,&~~\#tests\,&\,time\,&\,$|\mi{SSA}|$\,&\,\#tests\,&\,time\, \rule{0pt}{2.8mm} \\ 
&               &            &  & ~(s.) & $\times 10^6$ &~(s.)& & &~(s.) \\ \hline
\,cmugigamax       &\,104 &\,2,007\,   &\,989    &\,4.5  &~~100 &\,1,649\,&~~50   &\,3,421\bm{^*} &\,22  \rule{0pt}{2.8mm}\\ \hline
\,kenoopp1        &\,162 &\,2,247\,  &\,1,007  &\,4.7  &~~100&\,2,147\,  &~~72  &\,2,345\bm{^*}   &\,20 \rule{0pt}{2.8mm} \\ \hline
\,pd..ackjack4\,&\,198 &\,1,251\,  &\,  216  &\,0.4   &~~2.4\bm{^*} \,&~30\,  &~~62   &\,728\bm{^*} &\,8.5\rule{0pt}{2.8mm}\\ \hline
\,pdtpmsns2        &\,236 &\,1,187\,  &\,967    &\,1.0  &~~100 &\,1,507\,&~~31 &\,560\bm{^*}&\,1.2 \rule{0pt}{2.8mm}\\ \hline
\,abp4pold         &\,258 &\,2,547\,   &\,1,306   &\,6.6 &~~100 &\,3,555\,&~~22 &\,500\bm{^*} &\,1.5  \rule{0pt}{2.8mm}\\ \hline
\,pdtvissfeistel\,&\,342 &\,2,111\,   &\,1,232   &\,2.8 &~~100&\,4,557\,&~~86  &\,932\bm{^*}       &\,35 \rule{0pt}{2.8mm}\\ \hline
\,neclaftp1001    &\,398 &\,2,707\,  &\,1,723       &\,0.2    &~~100&1,313&~~40    &\,417\bm{^*}    &\,6.7 \rule{0pt}{2.8mm}\\ \hline
\,mentorb..00\,   &\,640&\,2.951\, &\,1,932   &\,10  \,&~~3.4\bm{^*} \,&\,688\,&~~43&\,1,476\bm{^*} &\,15   \rule{0pt}{2.8mm}\\ \hline
\end{tabular}                
\end{center}
\label{tbl:flt_tst}
\end{table}

%% file: f10ootnote.tex
To generate SAF tests we wrote a simple program based on
Minisat. Although this program did not have the efficiency of a
dedicated ATPG tool it was good enough for the purpose of studying the
effectiveness of SAF tests.

%% file: e8xper_corner_cases.tex
\subsection{Testing corner cases}
\label{ssec:ecorners}

In the third experiment, we generated \acts{s} and \aacts{s} to test
corner cases (see Subsection~\ref{ssec:corners}). First, we formed a
circuit $K$ that evaluates to 0 for almost all input assignments.  So,
the assignments for which $K$ evaluates to 1 are corner
cases\footnote{We assume here that $K$ is a subcircuit of some circuit
  $M$. The input assignments for which $K$ evaluates to 1 are corner
  cases for $M$.}.  We compared the frequency of hitting corner cases
by random tests and by tests of a set $T$ built by \pct as
follows. Let $N$ be a miter of copies $K'$ and $K''$ (see
Figure~\ref{fig:gen_miter}). The set $T$ was generated using a
projection of $N$ either on the set $X$ of input variables or an
internal cut of $N$.

%
\input{a5ndify.tbl}
%

To build circuit $K$, we extracted the circuit $R$ specifying the next
state function of a latch of a HWMCC-10 benchmark and composed it with
an $n$-input AND gate as shown in Figure~\ref{fig:corners}.  The
circuit $K$ outputs 1 only if $R$ evaluates to 1 and the first
$n\!-\!1$ inputs variables of the AND gate are set to 1 too. So the
input assignments for which $K$ evaluates to 1 are ``corner cases''.

\input{c4orners.fig}

The results of the experiment are given in Table~\ref{tbl:corners}.
The first two columns name the benchmark and latch whose next state
function was used as circuit $R$. The next three columns give the
total number of input variables of $K$, the value of $n$ in the
$n$-input AND gate fed by $R$ and the number of gates in circuit
$K$. The following pair of columns describes the performance of random
testing. The first column of this pair gives the total number of
tests. The next column shows the percentage of times circuit $K$
evaluated to 1 (and so a corner case was hit). The last five columns
of Table~\ref{tbl:corners} describe the results of \pct. The first
column of the five indicates whether a \acts or \aacts was
generated. The second column gives the size of set $V$ on which the
projection of $N$ was computed.  \acts{s} were generated with $V = X$.
When computing \aacts{s}, the set $V$ formed an internal cut of $N$
and the value of $\mi{Tries}$ was set to 1.  The next column shows the
size of the test set. The fourth column gives the percentage of times
a corner case was hit. The last column shows the total run time.

The examples of Table~\ref{tbl:corners} were generated in pairs that
shared the same circuit $R$ and were different only in the size of the
AND gate fed by $R$. For instance, in the first and second entry of
Table~\ref{tbl:corners}, circuit $K$ was obtained by composing the
same circuit $R$ extracted from benchmark \ti{pdtvisgigamax5} with
10-input and 30-input AND gates respectively. Table~\ref{tbl:corners}
shows that for circuits with a 10-input AND gate, random testing hit
corner cases but the percentage of those events was much lower than
for \acts{s} and \aacts{s}. On the other hand, 100 millions of random
tests failed to hit a single corner case for examples with a 30-input
AND gate in sharp contrast to \acts{s} and \aacts{s}.

%% file: a5ndify.tbl.tex
%
%
\begin{table}[h!]
\small
\caption{\ti{Testing corner cases}}
\vspace{-5pt}
\scriptsize
\begin{center}
\begin{tabular}{|@{~}l@{~}|@{~}l@{~}|@{~}l@{~}|@{~}l@{~}|@{}l@{}|@{}l@{}|@{}l@{}|@{}l@{}|@{}l@{}|@{}l@{}|@{}l@{}|@{}l@{}|@{}l@{}|} \hline
~~name          &la- & \#inp &\#and &\,\#ga- & \multicolumn{2}{c|}{random} & \multicolumn{5}{c|}{testing by}  \\
                &tch&  vars   &vars &\,tes  & \multicolumn{2}{c|}{testing} & \multicolumn{5}{c|}{\acts and \aacts}  \\ \cline{6-12}
                &     &    &     &      &\,\#te-\,   &\,\#hits\,&\,test\, &        &~\#te- &\,\#hits\,  &~time\,  \\ 
                &     &    &     &      &\,sts\,     & ~\%    &\,set      &\,$|V|$\,&~sts   &~ \%      &~(s.) \\ \hline
  pd..gigamax5  & L46 & 43 & 10  &\,512\,  &\,$10^5$\,&~0.02  &\,\tiny{\acts}\,&~43 &~547   &~7.1   & ~0.2   \rule{0pt}{2.8mm}     \\ \hline
  pd..gigamax5  & L46 & 63 & 30  &\,512\,  &\,$10^8$\,&~0     & \,\tiny{\acts}\,&~63 &~1,243\,   &~3.0    &~ 0.2     \rule{0pt}{2.8mm}     \\ \hline
  pdtvisbpb1    & L48 & 46 & 10  &\,108\, &  $10^5$   &~0.04   &\,\tiny{\acts}\,&~46 &~398      &~9.0   &~0.01 \rule{0pt}{2.8mm}            \\ \hline
  pdtvisbpb1     & L48 & 66 & 30  &\,108\, &  $10^8$   &~0     & \,\tiny{\acts}\,&~66 &~736      &~3.1    &~0.03 \rule{0pt}{2.8mm}            \\ \hline
   abp4pold      &L270 & 139  & 10  &\,637\, & $10^5$  &~0.02  & \,\tiny{\aacts}\,&~35 &~2,047   &~8.5    &~0.9  \rule{0pt}{2.8mm}            \\ \hline
    abp4pold     &L270 & 159  & 30  &\,637\,& $10^8$   &~0     & \,\tiny{\aacts}\,&~55 &~5,256   &~3.3    &~2.1  \rule{0pt}{2.8mm}            \\ \hline
mentorbm1p00  &  L8670 &636  & 10 &\,1,630\, & $10^5$ &~0.1    & \,\tiny{\aacts}\,&~35 &~594     &~11     &~3.7 \rule{0pt}{2.8mm}            \\ \hline
mentorbm1p00   & L8670 &656  & 30 &\,1,630\, & $10^8$ &~0      & \,\tiny{\aacts}\,&~55 &~2,009   &~4.7    &~8.7 \rule{0pt}{2.8mm}            \\ \hline
\end{tabular}                
\end{center}
\label{tbl:corners}
\end{table}

%% file: c4orners.fig.tex
\setlength{\intextsep}{4pt}
\begin{wrapfigure}{L}{1in}
 \begin{center}
   \includegraphics[width=0.8in,height=1.3in]{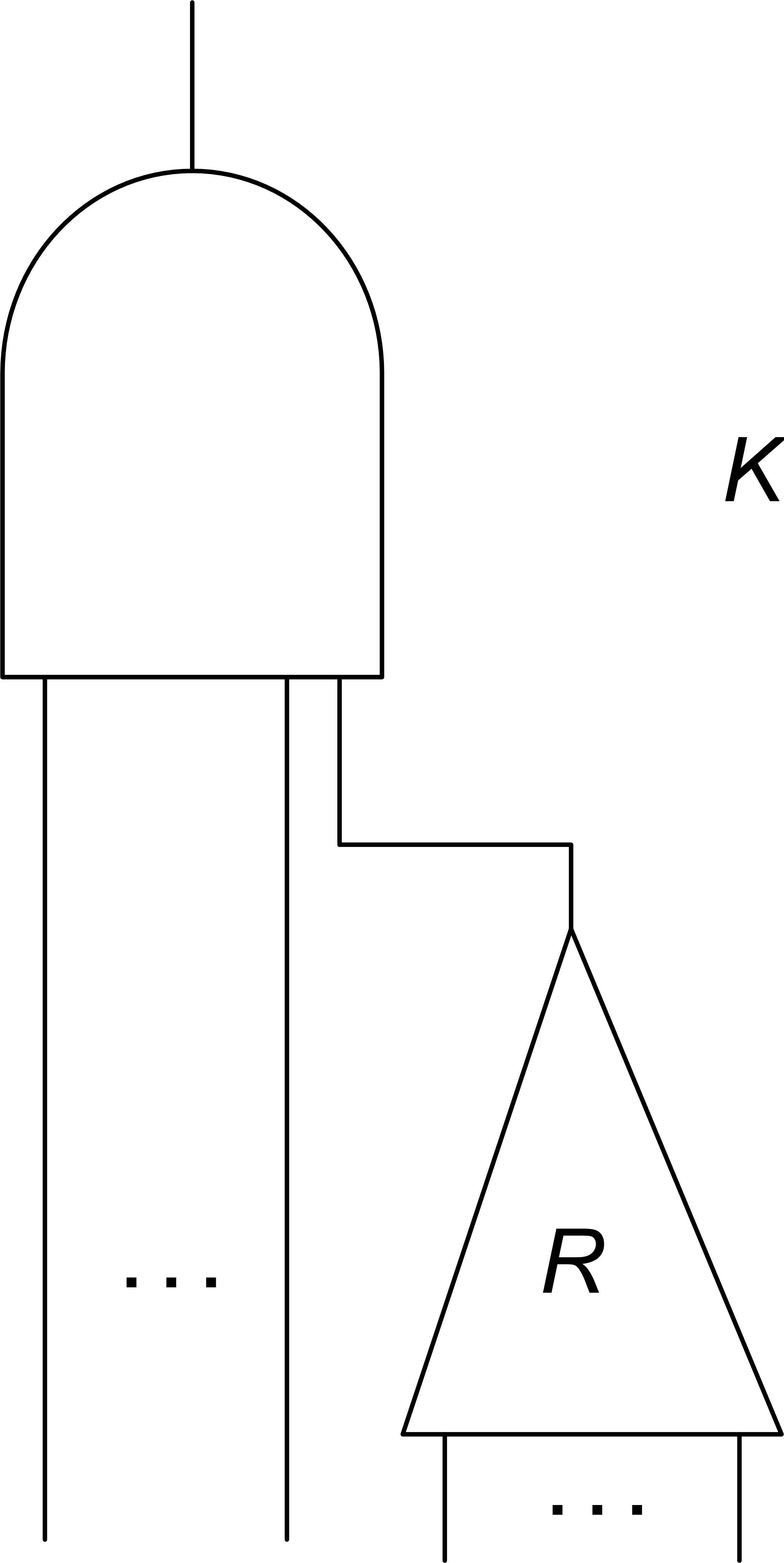}
  \end{center}
\vspace{-10pt}
\caption{Circuit $K$ whose output value is biased to 0}
\vspace{5pt}
\label{fig:corners}
\end{wrapfigure}

%% file: e9xper_missed_props.tex
\subsection{Testing properties defined incorrectly}
\label{ssec:missed_props}
%
\input{c8nt.tbl}

The objective of the last experiment was to show that one can use
tests generated by \ti{GenPCT} to address the problem of incorrect
definition of properties (see
Subsection~\ref{ssec:incomp_spec}). Assume that our design $D$
contains a $p$-bit counter. Denote this counter as $C_p$. Let
$\mi{val}(C_p)$ denote the current value of $C_p$. Assume that $C_p$
has one initial state $\mi{val}(C_p)=0$ and the counter is reset to
this state every time it reaches state $\mi{val}(C_p)=2^p\!-\!2$. So
property $\xi$ that $\mi{val}(C_p) < 2^p\!-\!1$ is true. Assume that
the correctness of $D$ requires proving a \ti{stronger} property
$\xi'$ that, say, $\mi{val}(C_p) < 2^p\!-\!3$ holds. In contrast to
$\xi$, property $\xi'$ \ti{is not true}. Let $N$ (respectively $N'$)
be a circuit specifying the correctness of $\xi$ (respectively $\xi'$)
for a given number of time frames. Let $T$ be a \aacts built by
\ti{GenPCT} when proving $N \equiv 0$ under the assumption that the
initial state of $C_p$ is $\mi{val}(C_p)=0$ (see
Subsection~\ref{ssec:seq_circ}). The idea of the experiment was to
show that $T$ could contain a counterexample to $N' \equiv 0$. So,
testing design $D$ by $T$ could expose a bug overlooked due to checking
property $\xi$ instead of $\xi'$.

Circuits $N$ and $N'$ above were obtained by unrolling the transition
relation of the counter $k$ times and adding gates specifying property
$\xi$ or $\xi'$. The value of $k$ was set to $2^p\!-\!2$ to guarantee
the existence of a trace breaking $\xi'$. Table~\ref{tbl:missed}
compares random and SAF tests\footnote{SAF tests were generated for
  the circuit obtained by unrolling the transition relation of the
  counter $k$ times. The input variables of this circuit corresponding
  to state varaibles of the first time frame were set to 0 (to take
  into account the initial state of the counter).} with a \aacts. The
latter was computed for a set $V$ specifying an internal cut of $N$
and $\mi{Tries}$ set to 1. The second column of Table~\ref{tbl:missed}
shows the number of times the transition relation was unrolled to
obtain $N$ and $N'$. The next two columns describe results of running
random and SAF tests. An asterisk marks test sets with a
counterexample to $\xi'$. The last three columns show the performance
of \aacts{s}. The first column of the three gives the number of
iterations needed to break $\xi'$. In each iteration, a new \aacts was
generated using a different choice of the center of SSA $P$ (see
Definition~\ref{def:ssa}). The last two columns give the total number
of tests and run time summed up over all iterations.

Table~\ref{tbl:missed} shows that random testing broke $\xi'$ only for
$C_4$ and $C_5$ and SAF tests succeeded only for $C_4$. Tests of
\aacts broke $\xi'$ for all 4 examples.

%% file: c8nt.tbl.tex
%
%
\begin{wraptable}{L}{1.7in}
\small
\caption{\ti{Testing a ``misdefined'' property of a counter. \aacts is
    computed for $|V|=14$}}
\vspace{-5pt}
\scriptsize
\begin{tabular}{|@{}l@{}|@{}l@{}|@{}l@{}|@{}l@{}|@{}l@{}|@{}l@{}|@{}l@{}|@{}l@{}|@{}l@{}|@{}l@{}|@{}l@{}|} \hline
           &\,\#time\,  &\,\#rand & & \multicolumn{3}{c|}{\aacts}  \\  \cline{5-7}
~$C_p$~  &\,fra-  &\,tests\,   &\,\#SAF\, &\,\#it\,& \,\#te-   &\,time\,  \\ 
      &\,mes      &$\times 10^6$ &\,tests   &\,ter & \,sts\,  &~(s.) \\  \hline
$\,C_4$  &~14   &\,0.001\bm{^*} &\,172\bm{^*}&\,1 & \,572\bm{^*}\,  & \,0.1 \rule{0pt}{2.8mm}    \\ \hline
$\,C_5$  &~30   &\,28\bm{^*}\,  &\,505 &\,2 & \,1,393\bm{^*}\,  & \,0.2  \rule{0pt}{2.8mm}   \\ \hline
$\,C_6$  &~62    &\,100    &\,1,276 &\,2 &\,338\bm{^*}\,  &\,0.5   \rule{0pt}{2.8mm}            \\ \hline
$\,C_7$  &~126   &\,100    &\,3,031 &\,1 &\,1,025\bm{^*}\,  &\,2.1   \rule{0pt}{2.8mm}            \\ \hline
\end{tabular}                
\label{tbl:missed}
\end{wraptable}

%% file: b2ackground.tex
\section{Background}
As we mentioned earlier, traditional testing checks if a circuit $M$
is correct as a whole. This notion of correctness means satisfying a
conjunction of \ti{many} properties of $M$. For this reason, one tries
to spray tests uniformly in the space of all input assignments.  To
improve the effectiveness of testing, one can try to run many tests at
once as it is done in symbolic simulation~\cite{SymbolSim}. To avoid
generation of tests that for some reason should be or can be excluded,
a set of constraints can be used~\cite{cnst_rand}. Another method of
making testing more reliable is to generate tests exciting a
particular set of events specified by a coverage
metric~\cite{cov_metr}. Our approach is different from those above in
that it is aimed at testing a particular property of $M$.

The method of testing introduced in~\cite{bridging} is based on the
idea that tests should be treated as a ``proof encoding'' rather than
a sample of the search space. (The relation between tests and proofs
have been also studied in software verification, e.g.
in~\cite{UnitTests,godefroid,Beckman}). In this paper, we take a
different point of view where testing becomes a \ti{part} of a formal
proof namely the part that performs structural derivations.

Reasoning about SAT in terms of random walks was pioneered
in~\cite{rand_walk}. The centered SSAs we introduce in this paper bear
some similarity to sets of assignments generated in de-randomization
of Sch\"oning's algorithm~\cite{balls}. Typically, centered SSAs are
much smaller than uncentered SSAs of~\cite{ssp}.

The first version of \sas procedure is presented in
report~\cite{cmpl_tst}.  It has a much tighter integration between the
structural part (computation of SSAs) and semantic part (derivation of
formula $H$ implied by the original formula). The advantage of the new
version of \sas described in this paper is twofold. First, it is much
simpler than \sas of~\cite{cmpl_tst}. In particular, any resolution
based SAT-solver that generates proofs can be used to implement the
new \sas.  Second, the simplicity of the new version makes it much
easier to achieve the level of scalability where \sas becomes
practical.


%% file: c6onclusion.tex
\section{Conclusion}
We consider the problem of finding a Complete Test Set (CTS) for a
combinational circuit $N$ that is a test set proving $N \equiv 0$.  We
use the machinery of stable sets of assignments to derive non-trivial
CTSs i.e. those that do not include all possible input assignments.
Computing a CTS for a large circuit $N$ is inefficient. So, we present
a procedure that generates a test set for a ``projection'' of $N$ on a
subset $V$ of variables of $N$. Depending on the choice of $V$\!, this
procedure generates a test set \acts that is an approximation of an
CTS or a test set \aacts that is an approximation of \acts. We give
experimental results showing that \aacts{s} can be efficiently
computed even for large circuits and are effective in solving
verification problems.

%% file: a4ppendix.tex
\vspace{15pt}
\appendices
\section{Proofs}
\label{app:proofs}
\setcounter{proposition}{0}
\begin{proposition}
 Formula $H$ is unsatisfiable iff it has an SSA.
\end{proposition}
\begin{proof} \tb{If part.}  Assume the contrary i.e. $P$ is an SSA of $H$ with
center \sub{p}{init} and AC-mapping \Fi and $H$ is
satisfiable. Let \pnt{s}~\,be an assignment satisfying $H$ that is the
closest to \sub{p}{init} in terms of the Hamming distance.  Then
procedure \ti{BuildPath} (see Fig.~\ref{fig:bld_path}) can build a
sequence of assignments $\ppnt{p}{1},\dots,\ppnt{p}{i}$ such that
\begin{itemize}
\item $i = \mi{Hamming\_distance}(\sub{p}{init},\pnt{s})+1$
\item $\ppnt{p}{1} = \sub{p}{init}$ and $\ppnt{p}{i} = \pnt{s}$
\end{itemize}
By definition of \ti{BuildPath}, assignment \ppnt{p}{j+1} is closer
to \pnt{s} and farther away from \sub{p}{init} than \ppnt{p}{j} where
$1 \leq j \leq i-1$. This means that \ppnt{p}{j+1} is in
$\mi{Nbhd}(\sub{p}{init},\ppnt{p}{j},C)$ where $C = \Fi(\ppnt{p}{j})$.
In particular, \pnt{s} is in
$\mi{Nbhd}(\sub{p}{init},\ppnt{p}{i-1},C)$ and so \pnt{s} is in
$P$. However, by definition of an SSA, $P$ consists only of
assignments falsifying $H$. Thus, we have a contradiction.

\vspace{4pt}
\noindent\tb{Only if part}. Assume that formula $H$ is unsatisfiable.  
By applying \ti{BuildSSA} (see Fig.~\ref{fig:bld_ssa}) to $H$, one
generates a set $P$ that is an SSA of $H$ with respect to some
center \sub{p}{init} and AC-mapping \Fi.
\end{proof}

\section{CTSs And Circuit Redundancy}
\label{app:red}
Let $N \equiv 0$ hold. Let $R$ be a cut of circuit $N$. We will denote
the circuit between this cut and the output of $N$ as $N_R$ (see
Figure~\ref{fig:cut}). We will say that $N$ is \tb{non-redundant} if
$N_R \not\equiv 0$ for any cut $R$ other than the cut specified by
primary inputs of $N$. Note that if $N_{R} \not\equiv 0$ for some cut
$R$, then $N_{R'} \not\equiv 0$ for \ti{every} cut $R'$ located above
$R$.

Definition~\ref{def:cts} of a CTS may not work well if $N$ is highly
redundant.  Assume, for instance, that $N_R \equiv 0$ holds for a cut
$R$. This means that the clauses specifying gates of $N$ below $R$
(i.e. those that are not in $N_R$) are redundant in $F_N \wedge
z$. Then one can build an SSA $P$ for $F_N \wedge z$ as follows.  Let
$P_R$ be an SSA for $F_{N_R} \wedge z$.  Let \pnt{v} be an arbitrary
assignment to the variables of $\V{N} \setminus \V{N_R}$.  Then by
adding \pnt{v} to every assignment of $P_R$ one obtains an SSA for
$F_N \wedge z$. This means that for any test \pnt{x}, \cube{x}
contains an SSA of $F_N \wedge z$. Therefore, according to
Definition~\ref{def:cts}, circuit $N$ has a CTS consisting of just one
test.

\input{c2ut.fig}

The problem above can be solved using the following observation. Let
$T$ be a set of tests \s{\ppnt{x}{1},\dots,\ppnt{x}{k}} for $N$ where
$k \leq 2^{|X|}$. Denote by $\vec{r}_i$ the assignment to the
variables of cut $R$ produced by $N$ under input \ppnt{x}{i}. Let
$T_R$ denote \s{\ppnt{r}{1},\dots, \ppnt{r}{k}}. Denote by $T^*_R$ the
set of assignments to variables of $R$ that cannot be produced in $N$
by any input assignment.  Now assume that $T$ is constructed so that
$T_R \cup T^*_R$ is a CTS for circuit $N_R$. This does not change
anything if $N_R$ is itself redundant (i.e. if $N_{R'} \equiv 0$ for
some cut $R'$ that is closer to the output of $N$ than $R$). In this
case, it is still sufficient to use $T$ of one test because $N_R$ has
a CTS of one assignment (in terms of cut $R$).  Assume however, that
$N_R$ is non-redundant. In this case, there is no ``degenerate'' CTS
for $N_R$ and $T$ has to contain at least $|T_R|$ tests.  Assuming
that $T^*_R$ alone is far from being a CTS for $N_R$, a CTS $T$ for
$N$ will consist of many tests.

So, one can modify the definition of CTS for a redundant circuit $N$
as follows.  A test set $T$ is a CTS for $N$ if there is a cut $R$
such that
\begin{itemize}
\item circuit $N_R$ is non-redundant i.e.
   \begin{itemize}
   \item[$\bullet$] $N_R \equiv0$ holds
   \item[$\bullet$] $N_R' \not\equiv 0$ for every cut $R'$ above $R$
   \end{itemize}
\item set $T_R \cup T^*_R$ is a CTS for $N_R$.
\end{itemize}

%% file: c2ut.fig.tex
\setlength{\intextsep}{4pt}
\begin{wrapfigure}{L}{1.3in}
\includegraphics[width=1in]{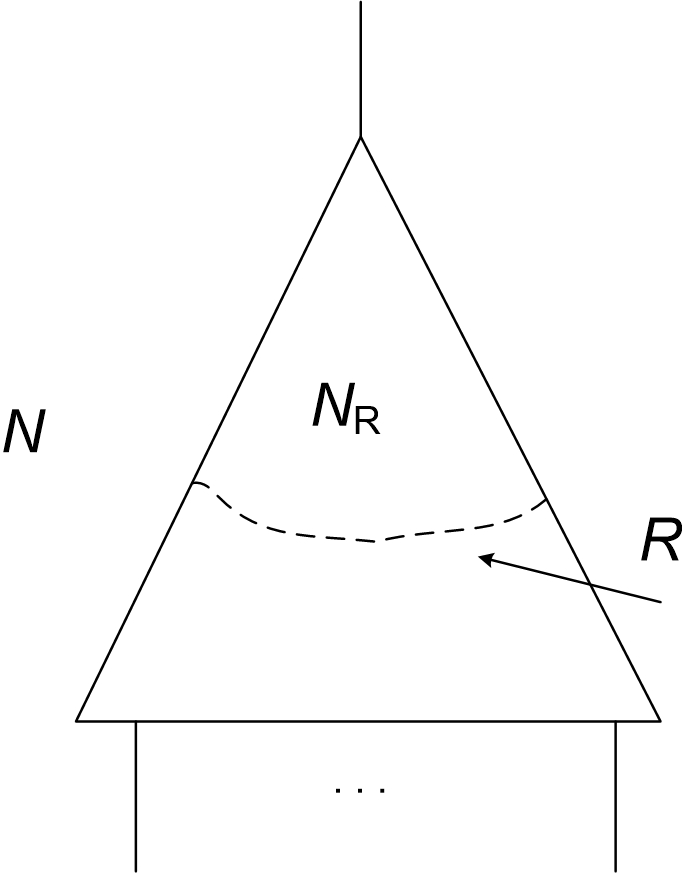}
\vspace{-10pt}
\caption{A cut $R$ in circuit $N$}
\label{fig:cut}
\end{wrapfigure}